\documentclass[a4paper,USenglish,cleveref, autoref]{lipics-v2019}

\usepackage{tikzit}
\usepackage{xspace}
\usepackage{stackengine}
\usepackage{amsthm}
\usepackage{amsmath}
\usepackage{listings}
\usepackage{xcolor}
\usepackage[noend]{algpseudocode}
\usepackage{algorithm}
\usepackage{xifthen}

\tikzstyle{ssource}=[fill=white, draw=black, shape=rectangle, minimum height=16pt]
\tikzstyle{ssink}=[fill={rgb,255: red,191; green,191; blue,191}, draw=black, shape=rectangle, minimum height=16pt]
\tikzstyle{sliteral}=[fill=white, draw=black, shape=regular polygon, regular polygon sides=3, inner sep=-0.9pt]
\tikzstyle{sclause}=[fill={rgb,255: red,191; green,191; blue,191}, draw=black, shape=circle]
\tikzstyle{sroot}=[fill=white, draw=black, shape=regular polygon, tikzit fill={rgb,255: red,191; green,0; blue,64}, regular polygon sides=7]
\tikzstyle{sgadgetframe}=[fill=none, draw=black, shape=rectangle, minimum width=3cm, minimum height=2.8cm]
\tikzstyle{normalnode}=[fill=white, draw=black, shape=circle, inner sep=0pt, minimum width=12pt]
\tikzstyle{big rightarrow}=[single arrow, draw=black, fill={black!10}, minimum height=.9cm]
\tikzstyle{node_new}=[fill=white, draw=black, shape=circle, inner sep=0pt, minimum width=10 pt]
\tikzstyle{shifted_node}=[draw=none, xshift=2pt]

\tikzstyle{normaledge}=[-]
\tikzstyle{normaledge-directed}=[->]
\tikzstyle{rootedge}=[-, draw=black, dashed]
\tikzstyle{rootedge-directed}=[->, draw=black, dashed]
\tikzstyle{explicit-edge}=[->]
\tikzstyle{implicit-edge}=[->, dashed]
\tikzstyle{thick_edge}=[-, thick]
\tikzstyle{dashed_edge}=[-, dashed]
\tikzstyle{gray}=[-, draw={rgb,255: red,191; green,191; blue,191}]
\tikzstyle{thick-directed}=[->, thick]
 
 \newtheorem{observation}{Observation}

\def\shft#1{\stackunder[65pt]{}{\kern-65pt #1}}

\newcommand{\sourcenode}[1]{\ensuremath{s_{#1}}\xspace}
\newcommand{\sinknode}[1]{\ensuremath{t_{#1}}\xspace}
\newcommand{\lit}[1]{\ensuremath{x_{#1}}\xspace}
\newcommand{\litnot}[1]{\ensuremath{\overline{x_{#1}}}\xspace}
\newcommand{\clause}[1]{\ensuremath{C_{#1}}\xspace}
\newcommand{\gadget}[1]{\ensuremath{V_{X_{#1}}}\xspace}

\newcommand{\minprimsu}{\textsc{ULGT}\xspace}
\newcommand{\kprimsu}{\textsc{k-ULGT}\xspace}
\newcommand{\minprimsd}{\textsc{DLGT}\xspace}
\newcommand{\kprimsd}{\textsc{k-DLGT}\xspace}

\newcommand{\NP}{\ensuremath{\mathsf{NP}}\xspace}
\newcommand{\SAT}{\ensuremath{\mathsf{SAT}}\xspace}
\newcommand{\APX}{\ensuremath{\mathsf{APX}}\xspace}

\newcommand{\oneton}{\ensuremath{[n]}\xspace}
\newcommand{\oneto}[1]{\ensuremath{[#1]}\xspace}
\newcommand{\uedge}[1]{\ensuremath{\{#1\}}\xspace}

\newcommand{\Eadditional}{\ensuremath{E_{\oplus}}\xspace}
\newcommand{\Eremoved}{\ensuremath{E_{\ominus}}\xspace}
\newcommand{\USF}{\textsc{USF}\xspace}
\newcommand{\FALGplus}{\ensuremath{F_{ALG,\oplus}}\xspace}
\newcommand{\FOPTplus}{\ensuremath{F_{OPT,\oplus}}\xspace}
\newcommand{\FALGminus}{\ensuremath{F_{ALG,\ominus}}\xspace}
\newcommand{\FOPTminus}{\ensuremath{F_{OPT,\ominus}}\xspace}

\title{On the Complexity of Local Graph Transformations}

\author{Christian Scheideler}{Paderborn University, Germany \and \url{https://cs.uni-paderborn.de/en/ti/}}{scheideler@upb.de}{}{}
\author{Alexander Setzer}{Paderborn University, Germany \and \url{https://cs.uni-paderborn.de/en/ti/}}{asetzer@mail.upb.de}{}{}
\authorrunning{C. Scheideler and A. Setzer}
\Copyright{C. Scheideler and A. Setzer}
\funding{This work was %
supported by the German Research Foundation (DFG) within the Collaborative Research Center "On-The-Fly Computing" (SFB 901) under Grant No.: GZ SFB 901/02.}

\hideLIPIcs
\nolinenumbers

\begin{document}
\maketitle
\begin{abstract}
We consider the problem of transforming a given graph $G_s$ into a desired graph $G_t$ by applying a minimum number of primitives from a particular set of \emph{local graph transformation primitives}.
These primitives are local in the sense that each node can apply them based on local knowledge and by affecting only its $1$-neighborhood.
Although the specific set of primitives we consider makes it possible to transform any (weakly) connected graph into any other (weakly) connected graph consisting of the same nodes, they cannot disconnect the graph or introduce new nodes into the graph, making them ideal in the context of supervised overlay network transformations.
We prove that computing a minimum sequence of primitive applications (even centralized) for arbitrary $G_s$ and $G_t$ is \NP-hard, which we conjecture to hold for any set of local graph transformation primitives satisfying the aforementioned properties.
On the other hand, we show that this problem admits a polynomial time algorithm with a constant approximation ratio.

This publication is the (revised) full version of a paper that appeared at ICALP'19 \cite{DBLP:conf/icalp/ScheidelerS19}.
\end{abstract}

\section{Introduction}

Overlay networks are used in many contexts, including peer-to-peer systems,
multipoint VPNs, and wireless ad-hoc networks. In fact, any distributed system
on top of a shared communication infrastructure usually has to form an overlay
network (i.e., its participating sites have to know each other or at least
some server) in order to allow its members to exchange information.

A fundamental task in the context of overlay networks is to maintain or adapt
its topology to a desired topology, where the desired topology might either be
pre-defined or depend on a certain objective function. The problem of reaching
a pre-defined topology has been extensively studied in the context of
self-stabilizing overlay networks (e.g., \cite{corona,JRSST09,DolevK08,AspnesW07,JacobRSS2012,DBLP:journals/tcs/BernsGP13}), and the problem of adapting the
topology based on a certain objective function has been studied in the context
of self-adapting and -optimizing overlay networks (e.g., \cite{DBLP:journals/ton/SchmidASBHL16, DBLP:conf/podc/FabrikantLMPS03, DBLP:conf/soda/AlbersEEMR06, DBLP:conf/wine/HaleviM07, DBLP:conf/podc/DemaineHMZ07, DBLP:journals/siamdm/AlonDHL13, DBLP:conf/sagt/Cord-LandwehrHKS12, DBLP:journals/topc/BiloGLP16}). Many of these approaches are
decentralized, and because of that, the work (in terms of number of edge
changes) they need to adapt to a desired topology might be far away from the
minimum possible work to reach that topology. In fact, no non-trivial results
on the competitiveness of decentralized overlay network adaptations are known
so far other than handling single join or leave operations, and it is
questionable whether any good competitive result can be achieved with a
decentralized approach. An alternative approach would be that a server is
available for controlling the network adaptations, and this has already been
considered in the context of so-called supervised overlay networks.

In a {\em supervised overlay network} there is a dedicated, trusted node
called {\em supervisor} that controls all network adaptations but otherwise is
not involved in the functionality of the overlay network (such as serving
search requests), which is handled in a peer-to-peer manner. This has the
advantage that even if the supervisor is down, the overlay network is still
functional. Solutions for supervised overlay networks have been proposed in \cite{DBLP:conf/ispan/KothapalliS05, DBLP:conf/ipps/FeldmannKSS18}, for example, and the results in \cite{DBLP:conf/ispan/KothapalliS05} imply that, for specific
overlay networks, any set of node arrivals and departures can be handled in a
constant competitive fashion (concerning the work needed for adding and
removing edges) to get back to a desired topology. But no general result is
known so far for supervised overlay networks concerning the competitiveness of
converting an initial topology into a desired topology. Also, no result is
known so far on how to handle the problem that a supervisor could be faulty or
even act maliciously.

A malicious supervisor would pose a significant problem for an overlay network
since it could easily launch \emph{Sybil attacks} (i.e., flooding the overlay
network with fake or adversarial nodes) or \emph{Eclipse attacks} (i.e.,
isolating nodes from other nodes in the overlay network). We thus ask: Can we
limit the power of a supervisor such that it cannot launch an eclipse or sybil
attack while still being able to convert the overlay network from any
connected topology to any other connected topology?

We answer the question to the affirmative by determining a set of graph
transformation commands, also called \emph{primitives}, that only the
supervisor may issue to the nodes. These primitives are powerful enough to
transform any (weakly) connected topology into any other (weakly) connected
topology but still allow the nodes to locally check that applying them does
not disconnect the network or introduce a new node into the network. We
additionally aim at minimizing the reconfiguration overhead, i.e., the number
of commands to be issued (and, related to this, the number of changes to be
made to node neighborhoods) to reach a desired topology. Unfortunately, as we
will show, this cannot be done efficiently for the set of primitives we
consider unless $\mathsf{P} \ne \NP$, and we conjecture that this holds for
any set of commands that has the aforementioned property of giving the
participants the ability to locally check that they cannot be used for eclipse
or sybil attacks. However, we are able to give an $O(1)$-approximation
algorithm for this problem.

\subsection{Model and Problem Statement}
We model the overlay network as a graph, i.e., nodes represent participants of the network and if there is a directed edge $(u,v)$ in the graph, this means that there is a connection from $u$ to $v$.
Undirected edges $\{u,v\}$ model the two connections from $u$ to $v$ and from $v$ to $u$.
Since there may be multiple connections between the same pair of participants, the graphs we consider in this work are multigraphs, i.e., edges may appear several times in the (multi-)set of edges.
For convenience throughout this work we will use the term ``graph'' instead of multigraph and refer to ``edge sets'' even though their elements need not be unique.

We consider the following set $P_d$ of four primitives for the manipulation of directed graphs, first introduced by Koutsopoulos et al.~\cite{KoutsopoulosSS17} in the context of overlay networks:
\begin{description}
\item[Introduction]  If a node $u$ has a reference of two nodes $v$ and $w$ with $v \neq w$, $u$ \emph{introduces} $w$ to $v$ if $u$ sends a message to $v$ containing a reference of $w$ while keeping the reference.
\item[Delegation] If a node $u$ has a reference of two nodes $v$ and $w$ s.t. $u,v,w$ are all different, then $u$ \emph{delegates} $w$'s reference of $v$ if $u$ sends a message to $v$ containing a reference of $w$ and deletes the reference of $w$.
\item[Fusion] If a node $u$ has two references $v$ and $w$ with $v=w$, then $u$ \emph{fuses} the two references if it only keeps one of these references.
\item[Reversal] If a node $u$ has a reference of some other node $v$, then $u$ \emph{reverses} the connection if it sends a reference of itself to $v$ and deletes its reference of $v$.
\end{description}

The four primitives are visualized in Figure~\ref{fig:primitives}.
Note that for the introduction primitive, it is possible that $w=u$, i.e., $u$ introduces itself to $v$.
To simplify the description, we sometimes say that a node $u$ introduces or delegates the \emph{edge $(u,v)$} if $u$ introduces $v$ to some other node or delegates $v$'s reference to some other node, respectively.

\begin{figure}[htb]
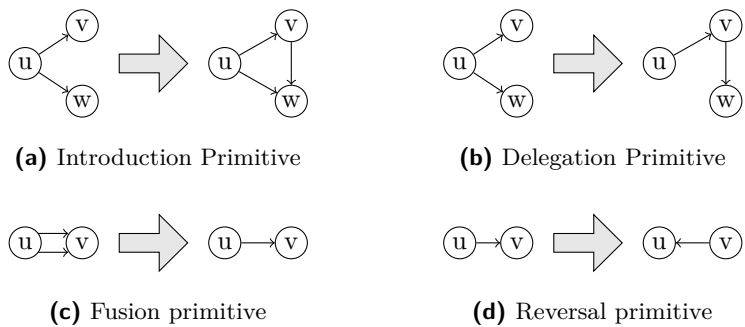

\centering
\captionsetup[subfigure]{aboveskip=0pt}
 \begin{subfigure}{.4\textwidth}
 \captionsetup{justification=centering}
  \ctikzfig{prim-introduction}
  \caption{Introduction Primitive}
  \label{fig:prim-intro}
 \end{subfigure}
 \begin{subfigure}{.4\textwidth}
 \captionsetup{justification=centering}
  \ctikzfig{prim-delegation}
  \caption{Delegation Primitive}
  \label{fig:prim-dele}  
 \end{subfigure}
 
 \hfill \\
 \hfill
  
 \begin{subfigure}{.4\textwidth}
 \captionsetup{justification=centering}
  \ctikzfig{prim-fusion}
  \caption{Fusion primitive}
  \label{fig:prim-fuse}   
 \end{subfigure}
\begin{subfigure}{.4\textwidth}
 \captionsetup{justification=centering}
 \ctikzfig{prim-reversal}
 \caption{Reversal primitive}
 \label{fig:prim-reverse}
\end{subfigure}
\caption{The four primitives in $P_d$ in pictures. %
}\label{fig:primitives}
\end{figure}

The primitives in $P_d$ are known to be universal (c.f.~\cite{KoutsopoulosSS17}), i.e., it is possible to transform any weakly connected graph into any other weakly connected graph by using only the primitives in $P_d$.
Note that for every edge $(u,v)$ used in any of the primitives, either $(u,v)$ still exists after the corresponding primitive is applied, or there is still an (undirected) path from $u$ to $v$ in the resulting graph.
This directly implies that no application of the primitives can disconnect the graph.
We assume that all connections are \emph{authorized}, meaning that both endpoints are aware of the other endpoint of this connection.
Thus, if for an edge $(u,v)$ that is supposed to be transformed into $(v,u)$ by an application of the reversal primitive, $v$ checks that $u$ actually was the previous endpoint of the former edge then the primitives cannot be used to introduce new nodes into the graph.

For undirected graphs, consider the set $P_u$ containing only the primitives introduction, delegation and fusion (defined correspondingly).
These three primitives, accordingly, are universal on undirected graphs, i.e., any connected undirected graph can be transformed into any other connected undirected graph by applying the primitives in $P_u$ (c.f.~\cite{KoutsopoulosSS17}).

We make the following observation:
\begin{observation}\label{obs:prims_basic_observation}
 The introduction primitive is the only primitive that can increase the number of edges in a graph.
 The fusion primitive is the only primitive that can decrease the number of edges in a graph.
 The delegation primitive is the only primitive that can remove the last edge between two nodes (i.e., an edge of multiplicity one).
\end{observation}

A \emph{computation} $C$ is a finite sequence $G_1 \Rightarrow G_2 \Rightarrow \dots \Rightarrow G_l$ of either directed or undirected graphs, in which each graph $G_{i+1}$ is obtained from $G_i$ by the application of a single primitive from $P_d$ or $P_u$, respectively.
The graphs $G_1$ and $G_l$ are called the \emph{initial} and the \emph{final} graphs of $C$, respectively.
The variable $l$ is called the \emph{length of the computation}.

We define the \emph{Undirected Local Graph Transformation Problem} (\minprimsu) as follows: given two connected undirected graphs $G_s, G_t$, find a computation of minimum length whose initial graph is $G_s$ and whose final graph is $G_t$.
The corresponding decision problem \kprimsu is defined as follows: given a positive integer $k$ and two connected undirected graphs $G_s$ and $G_t$, decide whether there is a computation with initial graph $G_s$ and final graph $G_t$ of length at most $k$.
Accordingly we define the \emph{Directed Local Graph Transformation Problem} (\minprimsd) and  \kprimsd, which differ from the according problems in that the graphs are directed.
\subsection{Related Work}\label{subsec:related_work}
Graph transformations have been studied in many different contexts and applications, including but not limited to pattern recognition, compiler construction, computer-aided software engineering, description of biological developments in organisms, and functional programming languages implementation (for a more detailed introduction and literature overview, we refer the reader to \cite{DBLP:journals/scp/AndriesEHHKKPST99}, \cite{DBLP:journals/entcs/Heckel06a}, or \cite{DBLP:conf/gg/1997handbook, DBLP:conf/gg/1999handbookVol3}).
Simply put, a graph transformation (or graph-rewriting) system consists of a set of rules $L \rightarrow R$ that may be applied to subgraphs isomorphic to $L$ of a given graph $G$ thus replacing $L$ with $R$ in $G$.
Since changing the labels assigned to a graph (graph relabeling) is also a kind of graph transformation, basically every distributed algorithm can be understood as a graph transformation system (c.f.~\cite{DBLP:conf/gg/1999handbookVol3}).
The type of graph transformations probably closest related to our work is the area of \emph{Topology Control} (TC).
In simple terms, the goal of TC is to select a subgraph of a given input graph that fulfills certain properties (such as connectivity) and optimizes some value (such as the maximum degree).
This problem has been studied in a variety of settings (for surveys on this topic see, e.g., \cite{DBLP:journals/pieee/LiLV13}, or \cite{DBLP:journals/comsur/AzizSFI13}) and although the usual approach is decentralized, there are also some centralized algorithms in this area (see, e.g., \cite{DBLP:conf/infocom/RamanathanH00}).
However, these works only consider the complexity of computing an optimal topology (instead of the complexity of transforming the graph by a minimum number of rule applications).
There is one work by Lin~\cite{DBLP:conf/isaac/Lin94} proving the \NP-hardness of the \emph{Graph Transformation Problem}, in which the goal is to find the minimum integer $k$ such that an initial graph $G_s$ can be transformed into a final graph $G_t$ by adding and removing at most $k$ edges in $G_s$.
Our work differs from that work in that we do not allow arbitrary edge relocations but restrict them to a set of rules that can be applied locally (and we also provide constant-factor approximation algorithms).

Our approximation algorithms use an approximation algorithm for the Undirected Steiner Forest Problem as a black-box (also known as the Steiner Subgraph Problem with edge sharing, or, in generalizations, the Survivable Network Design Problem or the Generalized Steiner Problem).
2-approximations of this problem were first given by Agrawal, Klein, and Ravi \cite{GSPAgrawal}, and by Goemans and Williamson \cite{GSPGoemans}, and later also by Jain~\cite{DBLP:journals/combinatorica/Jain01}.
Gupta and Kumar~\cite{Gupta:2015:GAS:2746539.2746590} showed a simple greedy algorithm to have a constant approximation ratio and recently, Gro{\ss} et al.~\cite{gro_et_al:LIPIcs:2018:8313} presented a local-search constant approximation for Steiner Forest.

\subsection{Our Contribution}
The main contributions of this paper are as follows:
We prove the Undirected and the Directed Local Graph Transformation Problem to be \NP-hard in \Cref{sec:nphardness}.
Furthermore, in \Cref{sec:approxalgos} we show that they belong to \APX, i.e., there exist constant approximation algorithms for these two problems.

\section{\NP-hardness results}\label{sec:nphardness}
In this section, we show the \NP-hardness of the Undirected Local Graph Transformation Problem by proving the \NP-hardness of \kprimsu (see \Cref{subsec:kprimsu_np_hard}).
Since \kprimsd's \NP-hardness is very similar for \kprimsu, we only briefly sketch the differences in \Cref{subsec:kprimsd_np_complete}.
Throughout this section, for any positive integer $i$ we use the notation $\oneto{i}$ to refer to the set $\{ 1, 2, \dots, i \}$.

\subsection{\kprimsu is \NP-hard}\label{subsec:kprimsu_np_hard}
We prove \kprimsu's hardness via a reduction from the Boolean satisfiability problem (\SAT) which was proven to be \NP-hard by Cook~\cite{Cook:1971:CTP:800157.805047} and, independently, by Levin~\cite{levin1973universal}.
We briefly recap \SAT as follows:
\begin{definition}[\SAT]
Given a set $X$ of $n$ Boolean variables $x_1, \dots, x_n$ and 
a Boolean formula $\Phi$ over the variables in $X$ in conjunctive normal form (CNF),
decide whether there is a truth assignment $t: X \rightarrow \{0,1\}$ that satisfies 
$\Phi$.
\end{definition}

To reduce \SAT to \kprimsu, we use the following reduction function:
\begin{definition}[Reduction function for $SAT \le_p \kprimsu$]\label{def:reduction_function_undirected}
 Let $S = (X,\Phi)$ be a \SAT instance, in which $X = \{x_1, \dots, x_n\}$ is the set of Boolean variables 
 and $\Phi = C_1 \land \dots \land C_m$ for clauses $C_1, \dots, C_m$.
 Then $f(S) = (G_s,G_t,k)$ in which $k = 2n + m$ and $G_s$ and $G_t$ are undirected graphs defined as follows.
 Without loss of generality, assume that each literal $y_i \in \{\lit{i}, \litnot{i}\}$ occurs only once in each clause.
 
 We define the following sets of nodes:
 $V_C = \{\clause{1}, \dots, \clause{m}\}$, and $\gadget{i} = \{ \lit{i}, \litnot{i}, \sourcenode{i}, \sinknode{i} \}$.
 Then, the set of nodes of $G_s$ and $G_t$ is $V = \bigcup_{1\le i\le n}V_{X_i} \cup V_C \cup \{r\}$. 
 For the set of edges, define $E_{X_i} = \{ \{s_i, \lit{i}\}, \{s_i, \litnot{i}\}, \{\lit{i}, t_i\}, \{\litnot{i}, t_i\}\}$ for every $i \in \oneton$, 
 $E_{C_j} = \{ \{ y_i, \clause{j} \} | y_i \in \{\lit{i},\litnot{i}\}\land y_i \text{ occurs in } \clause{j} \}$  for every $j \in \oneto{m}$,
 $E_{sr} = \{ \{s_i, r\} | 1 \le i \le n \}$, 
 $E_{tr} = \{ \{t_i, r\} | 1 \le i \le n \}$, $E_{Cr} = \{ \{C_j, r\} | 1 \le j \le m \}$.
 Both $G_s$ and $G_t$ have the edges in $\bigcup_{1\le i\le n}E_{X_i} \cup \bigcup_{1\le j\le m}E_{C_j}$.
 Additionally, $G_s$ has the edges in $E_{sr}$ and $G_t$ has the edges in $E_{tr} \cup E_{Cr}$. 
\end{definition}
Intuitively, each variable $x_i$ is mapped to a \emph{gadget} $X_i$ consisting of the four nodes $\lit{i}, \litnot{i}, \sourcenode{i}$, and $\sinknode{i}$.
Also each clause $C_j$ is connected with each literal occurring within it.
Lastly, in $G_s$, each of the $s_i$ is connected with the node $r$, whereas in $G_t$, each of the $t_i$ and each of the $C_j$ are connected with $r$.
\Cref{fig:reduction-graph-example} shows an example of the output of the reduction function for a given formula in CNF.

\begin{figure}
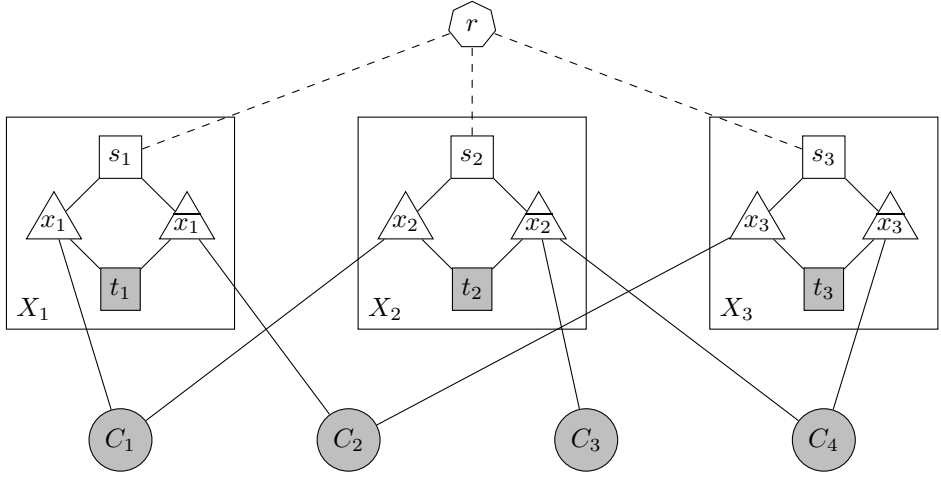

\ctikzfig{reduction-graph-example}
\caption{Graph $G_s$ returned by the reduction function in the \textbf{undirected} case for the (example) Boolean formula $(\lit{1} \lor \lit{2}) \land (\litnot{1} \lor \lit{3}) \land (\litnot{2}) \land (\litnot{2} \lor \litnot{3})$.
$G_t$ differs from $G_s$ in that the dashed edges do not exist and all grey nodes share an edge with node $r$.
}\label{fig:reduction-graph-example}
\end{figure}
We now show that every \SAT instance $S$ is satisfiable if and only if $f(S)$ is a ``yes'' instance of \kprimsu.
We start with the ``only if'' part for this is the simpler direction:
\begin{lemma}\label{lem:np:only_if}
 If a \SAT instance $S$ as in \Cref{def:reduction_function_undirected} is satisfiable then $f(S) = (G_s, G_t, k)$ with $k = 2n+m$ is a \kprimsu instance and there is a computation with initial graph $G_s$ and final graph $G_t$ of length at most $2n+m$.
\end{lemma}
\begin{proof}
 Assume there is a satisfying truth assignment $t: X \rightarrow \{0,1\}$ of $S$.
 For every $1 \le i \le n$ let $y_i := \lit{i}$ if $t(x_i) = 1$ or $y_i := \litnot{i}$ if $t(x_i) = 0$.
 We construct the following computation with initial graph $G_s$ and final graph $G_t$:
 \begin{enumerate}
  \item For every $1 \le i \le n$, $s_i$ delegates the edge $\{s_i, r\}$ to $y_i$.
  \item For every $C_j \in \{C_1, \dots, C_m\}$ choose one neighbor $z_j \in \{y_1, \dots, y_n\}$ (we show below that this exists), and let $z_j$ introduce $r$ to $C_j$.
  \item For every $1 \le i \le n$, $y_i$ delegates the edge $\{y_i, r\}$ to $t_i$.
 \end{enumerate}
Obviously, the length of this computation is $2n+m$.
To prove the missing part, recall that every $C_j$ is satisfied under $t$, i.e., there is at least one literal $z_j$ in $C_j$ that evaluates to true, i.e., there is an $i \in \oneton$ such that $z_j = \lit{i}$ if $t(\lit{i})= 1$, or $z_j = \litnot{i}$ if $t(\lit{i}) = 0$.
By definition of $y_i$, $z_j = y_i$.
Thus because $z_j$ occurs in $C_j$, $y_i$ was a neighbor of $C_j$ during Step~2.
\end{proof}

The ``if'' part is more complex.
We begin with the following insight that will prove helpful in the course of this part.
\begin{lemma}\label{lem:decomposition_passive_nodes}
 Suppose the nodes in the initial graph of a computation $C$ can be decomposed into disjoint sets $V_1, \dots, V_k, P$ such that there is no edge \uedge{u,v} for any $u \in V_i$, $v \in V_j$, $i,j \in \oneto{k}$, $i \ne j$ and throughout $C$ none of the nodes in $P$ applies a primitive.
 Then there is no edge \uedge{u,v} for any $u \in V_i$, $v \in V_j$, $i,j \in \oneto{k}$, $i \ne j$ in any graph of the computation.
\end{lemma}
\begin{proof}
 Assume there is a computation $C$ and sets $V_1, \dots, V_k, P$ as defined above and assume for contradiction that the claim is not true.
 We consider the first edge \uedge{u,v} such that $u \in V_i$, $v \in V_j$, $i,j \in \oneto{k}$, $i \ne j$.
 Clearly, it cannot have been created by the application of a fusion primitive.
 Thus it must have been created by an introduction or delegation primitive applied by a node $w$ that knew both $u$ and $v$ before the application of this primitive.
 This implies $w \in P$.
 However, the nodes in $P$ do not apply any primitives by assumption, which yields a contradiction.
\end{proof}

In the rest of this section, we use the following notation:
For a given instance of \kprimsu $(G_s, G_t, k)$, we say a computation is \emph{feasible} if and only if its initial graph is $G_s$, its final graph is $G_t$ and its length is at most $2n+m$.
 Furthermore, we say that the edge that is established during the application of an introduction or delegation primitive (the edge $(v,w)$ in \Cref{fig:prim-intro} and \Cref{fig:prim-dele}) is the \emph{result} of the introduction or delegation, respectively.

The next lemma we show represents a main building block of the proof of the ``if'' part.
\begin{lemma}\label{lem:undirected_yi}
 Let $S$ be a \SAT instance and let $(G_s, G_t, k) = f(S)$.
 For every computation $C$ with initial graph $G_s$ and final graph $G_t$ of length at most $2n+m$ it holds:
 There are $y_1, \dots, y_n$, $y_i \in \{\lit{i}, \litnot{i}\}$ for every $i \in \oneton$, such that in $C$ there are no edges other than $E(G_s) \cup E(G_t) \cup \{ \uedge{y_i, r} | i \in \oneton \}$ and no edge occurs twice in any graph of $C$ (where $E(G_s)$ and $E(G_t)$ denote the edge sets of $G_s$ and $G_t$, respectively).
\end{lemma}
The general idea of the proof of \Cref{lem:undirected_yi} is the following:
 To obtain the final graph, for each $j \in \oneto{m}$ the edge \uedge{C_j, r} has to be created and for each $i \in \oneton$ the edge \uedge{t_i, r} has to be created.
 Each of these creations involves a distinct application of a primitive.
 Therefore, only $n$ applications of primitives are left in a feasible computation.
 We show that the nodes in each gadget $i$ have to apply at least one primitive $p_i$ that does not create one of the above edges.
 This implies that each gadget may apply no other primitive than $p_i$ to create an edge that is not in the final graph and that the nodes $r$ and $C_j$ themselves cannot apply any primitives at all which by \Cref{lem:decomposition_passive_nodes} means that there are no inter-gadget edges.
 We use these facts to prove that $p_i$ is used to remove the edge \uedge{s_i, r} thereby creating either \uedge{\lit{i}, r} or \uedge{\litnot{i}, r}.
 
We split the full proof of \Cref{lem:undirected_yi} into three claims, which we prove individually.
The first claim we show is the following:
\begin{claim}\label{clm:claim1}
 In every feasible computation for every $i \in \oneton$, there is a node $z_i \in \{s_i, \lit{i}, \litnot{i}\}$ that applies a fusion primitive or an introduction or delegation primitive whose result is not in 
 $E_c := \{ \uedge{C_j, r} | 1 \le j \le m \} 
 \cup \{ \uedge{t_i, r} | 1 \le i \le n \} 
 \cup \{ \uedge{C_j, C_l} | i,j \in \oneto{m} \} 
 \cup \{ \uedge{t_i, C_j} | i \in \oneton, j \in \oneto{m} \}
 \cup \{ \uedge{t_i, t_k} | i, k \in \oneton \}$.
\end{claim}
\begin{proof}
 Assume for contradiction that there is a feasible computation and an $i \in \oneton$ such that there is no $z_i \in \{s_i, \lit{i}, \litnot{i}\}$ that applies a fusion primitive or an introduction or delegation primitive whose result is not in $E_c$.
 Note that $s_i$ cannot delegate away any of its incident edges \uedge{s_i, \lit{i}} or \uedge{s_i, \litnot{i}} as the result would not be in $E_c$ (for it would be incident to \lit{i} or \litnot{i}).
 Similarly, \lit{i} and \litnot{i} could not delegate away \uedge{\lit{i}, s_i} and \uedge{\litnot{i}, s_i}, respectively.
 Therefore, the edges \uedge{\lit{i}, s_i} and \uedge{s_i, \litnot{i}} must be kept throughout the computation.
 Now observe that $s_i$ has an initial degree of three.
 Since $s_i$ has a degree of two in the final graph, there must be at least one application of a primitive in which $s_i$'s degree decreases.
 Consider the last such application, i.e., the resulting neighborhood of $s_i$ is \lit{i} and \litnot{i} (remember that these edges persist throughout the computation).
 If it was an application of a fusion primitive, then \lit{i}, $s_i$ or \litnot{i} must have applied this primitive, yielding a contradiction.
 Otherwise, it must have been a delegation primitive applied by $s_i$ for no other primitive could reduce $s_i$'s degree then.
 However, the result of this delegation primitive must be an edge incident to \lit{i} or \litnot{i} then, i.e., an edge not in $E_c$ yielding a contradiction in this case as well.
 All in all, we have proven \Cref{clm:claim1}.
\end{proof}
The second claim we prove is: 
\begin{claim}\label{clm:claim2}
For every feasible computation $C$, the following holds:
\begin{enumerate}[(i)]
 \item \label{item:claim2:5} for each $i \in \oneton$ the nodes in $\gadget{i}$ may apply at most one primitive whose result is not an edge \uedge{t_k, r} or \uedge{C_j, r} for $k \in \oneton, j \in \oneto{m}$,
 \item \label{item:claim2:6} the nodes in $\{ r \} \cup \{ t_i | i \in \oneton \} \cup \{ C_j | j \in \oneto{m} \}$ do not apply any primitives at all, and
 \item \label{item:claim2:3} there is no graph in $C$ that contains an edge \uedge{u,v} such that $u \in \gadget{i}$ and $v \in \gadget{k}$ for any $i,k \in \oneton$ such that $i \ne k$.
\end{enumerate}
\end{claim}
\begin{proof}
 For the proof of this claim we distinguish between three types of edges.
 Edges \uedge{t_i, r} for some $i \in \oneton$ belong to type A.
 Edges \uedge{C_j, r} for some $j \in \oneto{m}$ belong to type B.
 Last every edge $e \notin E_c$ belongs to type $C$.
 Note that in order to obtain the final graph, the nodes need to establish $n$ edges of type A, $m$ edges of type B, and, according to \Cref{clm:claim1}, for every $i \in \oneton$, one of the nodes in $\{ s_i, \lit{i}, \litnot{i} \}$ has to apply a primitive whose result is an edge of type C or which is a fusion primitive.
 Since there may be at most $2n + m$ applications of primitives, no other primitives may be applied.
This immediately proves (\ref{item:claim2:5}).
 Since $r$ cannot create the first instance of an edge of type A or B by itself, $r$ cannot apply any primitive at all as this would require more than $2n + m$ primitive applications then ($n+m$ are used to create the type A/B edges and an additional $m$ primitives are applied by the nodes from $\{ s_i, \lit{i}, \litnot{i} | i \in \oneton \}$).
 For the same reason, the nodes in $\{ t_i | i \in \oneton \} \cup \{C_j | j \in \oneto{m} \}$ cannot apply any primitives, which together with the fact that $r$ cannot apply any primitives equates to (\ref{item:claim2:6}). 
 In addition this implies that the initial graph can be decomposed into $\gadget{1}, \gadget{2}, \dots, \gadget{n}, P$ with $P = \{ r \} \cup \{ C_j | j \in \oneto{m} \}$ such that there is no edge \uedge{u,v} for any $u \in \gadget{i}, v \in \gadget{k}$, $i,k \in \oneton$ and $i \ne k$, and the nodes in $P$ do not apply any primitive at all.
 Thus, (\ref{item:claim2:3}) follows from Lemma~\ref{lem:decomposition_passive_nodes}.
\end{proof}
 
 The last claim we show is the following:
\begin{claim}\label{clm:claim3}
 In every feasible computation, for every $i \in \oneton$ 
 there must be a graph containing an edge \uedge{\lit{i}, r} or \uedge{\litnot{i}, r}.
\end{claim}
\begin{proof} 
 To prove \Cref{clm:claim3}, we show that $s_i$ has to delegate \uedge{s_i, r} to \lit{i} or \litnot{i}.
 We do so by proving that $s_i$ cannot have a neighbor other than \lit{i}, \litnot{i} or $r$, which is sufficient because $r$ does not apply any primitive according to (\ref{item:claim2:6}) of \Cref{clm:claim2}.
 Assume for contradiction that $s_i$ has an edge \uedge{s_i, v} such that $v \notin \{\lit{i}, \litnot{i}, r \}$ and let \uedge{s_i, v} be the last such edge that occurs in a graph in the computation.
 Due to (\ref{item:claim2:3}) of \Cref{clm:claim2}, $v \in \{t_i | i \in \oneton\} \cup \{ C_j | j \in \oneto{m} \}$.
 Consider the node $w$ that applied an introduction or delegation primitive whose result was the edge $\uedge{s_i, v}$.
 For this to be possible, there must have been an edge \uedge{w, s_i} when $w$ applied the primitive.
 Since $r$, all $t_i$ for $i \in \oneton$, and all $C_j$ for $j \in \oneto{m}$ do not apply any primitives according to (\ref{item:claim2:6}) of \Cref{clm:claim2} and because of (\ref{item:claim2:3}) of \Cref{clm:claim2}, $w \in \{ \lit{i}, s_i, \litnot{i} \}$.
 According to (\ref{item:claim2:5}) of \Cref{clm:claim2}, the computation may not contain another application (than this one) of a delegation / introduction primitive by a node in \gadget{i} whose result is not an edge \uedge{t_k, r} or \uedge{C_j, r} for $k \in \oneton, j \in \oneto{m}$ (*).
 Since $\uedge{s_i, v} \notin E(G_t)$, this edge must be removed by some application of a delegation primitive.
 Note that since $v \in \{t_i | i \in \oneton \} \cup \{ C_j | j \in \oneto{m} \}$, $s_i$ must apply this primitive.
 Since \uedge{s_i, v} is the last edge different from \uedge{s_i, r}, \uedge{s_i, \lit{i}}, and \uedge{s_i, \litnot{i}}, $v$ must be delegated to one of the nodes $r$, \lit{i}, and \litnot{i}.
 If $s_i$ delegates \uedge{s_i, v} to \lit{i} or \litnot{i}, then the result is an edge \uedge{\lit{i}, v} or \uedge{\litnot{i}, v}, which contradicts (*).
 Thus assume $s_i$ delegates \uedge{s_i, v} to $r$.
 Note that after this delegation, the edge \uedge{s_i, r} exists and $s_i$ does not have a neighbor $v' \in \{t_i | i \in \oneton \} \cup \{ C_j | j \in \oneto{m} \}$ in any of the subsequent graphs (recall that $v$ was the last edge of its kind).
 Since $\uedge{s_i, r} \notin E(G_t)$, and $r$ does not apply any primitives according to (\ref{item:claim2:6}) of \Cref{clm:claim2}, this edge must be delegated to either \lit{i} or \litnot{i} yielding an edge \uedge{\lit{i}, r} or \uedge{\litnot{i}, r}, which contradicts (*) as well.
 As mentioned above this proves that $s_i$ has to delegate \uedge{s_i, r} to \lit{i} or \litnot{i}, and, as argued before as well, also implies the claim of the lemma.
\end{proof}

 \Cref{clm:claim3} gives that during a feasible computation, the edges \uedge{t_i, r} for all $i \in \oneton$, \uedge{C_j, r} for all $j \in \oneto{m}$ and \uedge{y_i, r}, $y_i \in \{ \lit{i}, \litnot{i} \}$ for all $i \in \oneton$ have to be created.
 Since each primitive can create at most one of these edges and the length of the computation is at most $2n + m$, this implies the claim of \Cref{lem:undirected_yi}.

The rest of the proof of the ``if'' part, as formalized by the following lemma, is comparably straightforward.
\begin{lemma}\label{lem:np:if} 
 Let $S$ be a \SAT instance as in \Cref{def:reduction_function_undirected}.
 If $f(S) = (G_s, G_t, k)$ with $k = 2n+m$ is a \kprimsu instance and there is a computation with initial graph $G_s$ and final graph $G_t$ of length at most $2n+m$, then $S$ is satisfiable. 
\end{lemma}
\begin{proof}
 Assume that $f(S) = (G_s, G_t, 2n+m)$ is a \kprimsu instance and there is a feasible computation $C$ for $f(S)$.
 According to \Cref{lem:undirected_yi} there are $y_1, \dots, y_n$, $y_i \in \{\lit{i}, \litnot{i}\}$ for every $i \in \oneton$ such that in $C$ there are no edges other than $E(G_s) \cup E(G_t) \cup \{ \uedge{y_i, r} | i \in \oneton \}$.
 Note that in $G_t$, for every $j \in \oneto{m}$ there is an edge \uedge{C_j, r} and each such edge must have been the result of an introduction or delegation primitive applied by an $y_i$, $i \in \oneton$ (as throughout $C$, the $C_j$ nodes do not have any other neighbors with an edge to $r$ that could possibly create this edge) .
 Let $g : \{ C_1, C_2, \dots, C_m \} \rightarrow \{ y_1, y_2, \dots, y_n \}$ be the mapping of each $C_j$ to the $y_i$ that applied a primitive that resulted in the edge \uedge{C_j, r}.
 Consider the truth assignment $t: X \rightarrow \{0,1\}$ such that $t(x_i) = 1$ if $y_i = \lit{i}$ and $t(x_i) = 0$ if $y_i = \litnot{i}$.
 Observe that $t(y_i) = 1$ for every $i \in \oneton$.
 Assume for contradiction that there is a clause $C_j$ in $S$ that does not evaluate to true under $t$.
 Note that $g(C_j)$ must occur in $C_j$ by construction.
 However, since $g(C_j) = y_i$ for some $i \in \oneton$ and $t(y_i) = 1$, we obtain the desired contradiction.
\end{proof}

Putting both parts together, \Cref{lem:np:only_if} and \Cref{lem:np:if} imply $\SAT \le_p \kprimsu$, from which we obtain:
\begin{corollary}
 \kprimsu is \NP-hard. 
\end{corollary}

\subsection{\kprimsd is \NP-hard}\label{subsec:kprimsd_np_complete}
The proof of the \NP-hardness of \kprimsd is very similar to that of \kprimsu.
Therefore, we do not restate the whole proof but point out the differences between the two proofs.

The reduction function is a ``directed version'' of \Cref{def:reduction_function_undirected}, in which each of the edges is assigned as unique direction.
Formally, it looks as follows:
\begin{definition}[Reduction function for $SAT \le_p \kprimsd$]\label{def:reduction_function_directed}
 Let $S = (X,\Phi)$ be a \SAT instance, in which $X = \{x_1, \dots, x_n\}$ is the set of Boolean variables 
 and $\Phi = C_1 \land \dots \land C_m$ for clauses $C_1, \dots, C_m$.
 Then $f(S) = (G_s,G_t,k)$ in which $k = 2n + m$ and $G_s$ and $G_t$ are undirected graphs defined as follows.
 Without loss of generality, assume that each literal $y_i \in \{\lit{i}, \litnot{i}\}$ occurs only once in each clause.
 
 We define the following sets of nodes: 
 $V_C = \{\clause{1}, \dots, \clause{m}\}$, and $\gadget{i} = \{ \lit{i}, \litnot{i}, \sourcenode{i}, \sinknode{i} \}$.
 Then, the set of nodes of $G_s$ and $G_t$ is $V = \bigcup_{1\le i\le n}V_{X_i} \cup V_C \cup \{r\}$. 
 For the set of edges, define $E_{X_i} = \{ (s_i, \lit{i}), (s_i, \litnot{i}), (\lit{i}, t_i), (\litnot{i}), t_i\}$ for every $i \in \oneton$, 
 $E_{C_j} = \{ ( y_i, \clause{j} ) | y_i \in \{\lit{i},\litnot{i}\} \land y_i \text{ occurs in } \clause{j} \}$ for every $j \in \oneto{m}$,
 $E_{sr} = \{ (s_i, r) | 1 \le i \le n \}$, 
 $E_{tr} = \{ (t_i, r) | 1 \le i \le n \}$, $E_{Cr} = \{ (C_j, r) | 1 \le j \le m \}$.
 Both $G_s$ and $G_t$ have the edges in $\bigcup_{1\le i\le n}E_{X_i} \cup \bigcup_{1\le j\le m}E_{C_j}$.
 Additionally, $G_s$ has the edges in $E_{sr}$ and $G_t$ has the edges in $E_{tr} \cup E_{Cr}$. 
\end{definition}
An example of the output of this reduction function is depicted in \Cref{fig:reduction-graph-example-directed}.
Note that this picture only differs from \Cref{fig:reduction-graph-example} in that the edges are directed.
More specifically, all solid edges are directed ``downwards'', the dashed edges are directed ``upwards'' and all grey nodes are supposed to have an an ``upward'' edge to $r$ in $G_t$.

\begin{figure}[ht]
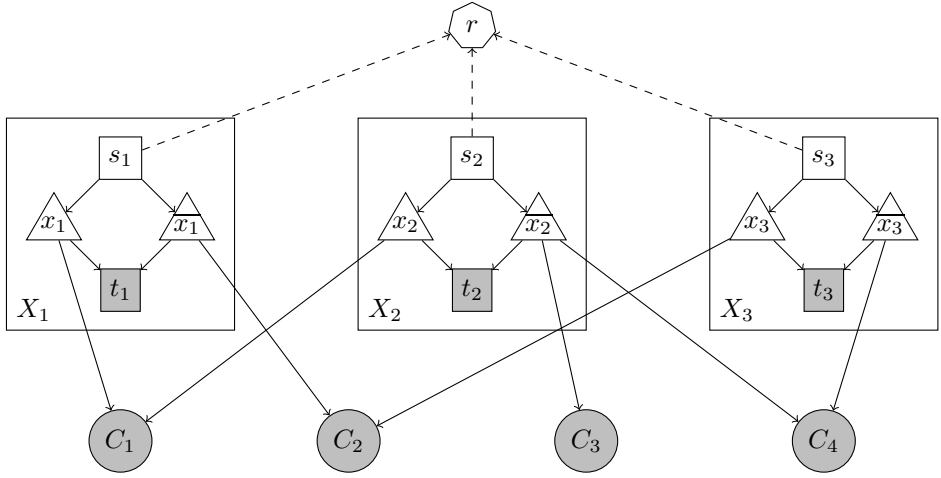

\ctikzfig{reduction-graph-example-directed}
\caption{Graph $G_s$ returned by the reduction function in the \textbf{directed} case for the (example) Boolean formula $(\lit{1} \lor \lit{2}) \land (\litnot{1} \lor \lit{3}) \land (\litnot{2}) \land (\litnot{2} \lor \litnot{3})$.
$G_t$ differs from $G_s$ in that the dashed edges do not exist and all grey nodes share an outgoing edge with node $r$.
}\label{fig:reduction-graph-example-directed}
\end{figure}

We now point out the differences in the proofs for the directed case, in which we refer to the lemmas from \Cref{subsec:kprimsu_np_hard}.
Note that to shorten the description, every edge $\{u,v\}$ that appears in the proof for the undirected case should be read as the directed edge $(u,v)$ unless noted differently.

\Cref{lem:np:only_if} (the ``only if'' part of the reduction) directly transfers to the directed case: the same approach described in that proof can be applied in the \kprimsd scenario as well.

For the directed version of \Cref{lem:decomposition_passive_nodes}, the same proof applies where the only additional argument to be mentioned is that the first edge $(u,v)$ such that $u \in V_i, v \in V_j, i, j \in \oneto{k}, i\ne j$ not only cannot have been created by the application of a fusion primitive, but also not by the application of a reversal primitive.

For the proof of the counterpart of \Cref{lem:undirected_yi}, we consider all three claims used for that lemma individually.
The argument showing that the claim of the lemma follows from these three claims is analogous.

\Cref{clm:claim1} must be completed such that $z$ applied either a fusion primitive or an introduction, delegation or reversal primitive whose result is not in $E_c$.
This set of edges is defined as: 
$E_c := \{ (C_j, r) | 1 \le j \le m \}
 \cup \{ (r, C_j) | 1 \le j \le m \} 
 \cup \{ (t_i, r) 1 \le i \le n \} 
 \cup \{ (r, t_i) 1 \le i \le n \}  
 \cup \{ (C_j, C_l) i,j \in \oneto{m} \} 
 \cup \{ (t_i, C_j) i \in \oneton, j \in \oneto{m} \}
 \cup \{ (C_j, t_i) i \in \oneton, j \in \oneto{m} \} 
 \cup \{ (t_i, t_k) i, k \in \oneton \}$.
In the proof, the assumption for contradiction includes that none of the $z_i$ applies a reversal primitive whose result is not in $E_c$.
In addition to that $s_i$ cannot delegate away $(s_i,\lit{i})$ and $(s_i,\litnot{i})$, we argue that $s_i$ cannot reverse this edge as the result would not be in $E_c$.
These two facts immediately imply that $(s_i, \lit{i})$ and $(s_i, \litnot{i})$ persist throughout the computation.
For the last primitive application that reduces $s_i$'s degree to two, we have to take into account that this could also be the application of a reversal primitive.
However, in that case, the result of this primitive application would be an edge whose one endpoint is $s_i$, i.e., an edge not in $E_c$, yielding a contradiction as well.

In the proof of \Cref{clm:claim2} we need the following additional argument to show that $r$ cannot apply any primitives (the previous argument that $r$ cannot create the first instance of an edge of type A or B by itself does not suffice if the reversal primitive could also be applied): 
Since no primitives other than the $n$ creating the type A edges, the $m$ creating the type $B$ edges and the $n$ creating the type $C$ edges or being a fusion primitive can be applied, there cannot be any edge $E_c \setminus \{(t_i,r), (C_j,r) | i \in \oneton, j \in \oneto{m} \}$ in any graph of the computation.
In particular, there can be no edge $(r,C_j)$ for any $j \in \oneto{m}$ or $(r,t_i)$ for any $i \in \oneton$.
Thus, $r$ cannot create the first instance of a type A or B edge via a reversal primitive and together with the previous argument we also obtain that $r$ cannot apply any primitive at all.

In the proof of \Cref{clm:claim3}, in order to show that $s_i$ has to delegate $(s_i,r)$ to $\lit{i}$ or $\litnot{i}$, we similarly show prove that $s_i$ cannot have an outgoing neighbor other than $\lit{i}$, $\litnot{i}$, or $r$.
In the directed scenario, however, this is not immediately sufficient because $r$ does not apply any primitive.
Additionally, we have to consider that $s_i$ cannot simply reverse edge $(s_i,r)$ as the resulting edge does not belong to $G_t$ and $r$ does not apply any primitive, which would be necessary to remove that edge again.
After that, note that the edge $(s_i,v)$ such that $v \notin \{\lit{i},\litnot{i},r\}$ which we assume to exist for contradiction cannot have been established by the application of a reversal primitive, for such a $v$ cannot apply any primitive at all according to \Cref{item:claim2:6} (since $v \in \{t_i | i \in \oneton\} \cup \{ C_j | j \in \oneto{m} \}$ follows equally in the directed case).
Thus it is feasible to continue with the consideration of the $w$ that applied an introduction or delegation primitive whose result was $(s_i,v)$.
The claim (*) in the directed case is that the computation may not contain another application (than the one applied by $w$ to create $(s_i,v)$) of a delegation, introduction, or reversal primitive by a node in $\gadget{i}$ whose result is not an edge $(t_k,r)$ or $(C_j,r)$ for $k \in \oneton, j \in \oneto{m}$.
Again, in order to remove $(s_i,v) \notin E(G_t)$, $s_i$ must apply a delegation primitive to remove this edge.
In this case, the argument is that $s_i$ cannot reverse the edge because $(v,s_i) \notin E(G_t)$ and $v$ does not apply any primitive (as argued before).
In the last contradiction of the proof, which relies on that $(s_i,r) \notin E(G_t)$, there is not only the option that $s_i$ delegates this edge to $\lit{i}$ or $\litnot{i}$, but also that $s_i$ reverses this edge.
This however would yield an edge $(r, s_i)$, which does not belong to $G_t$ and could not be removed, for $r$ does not apply any primitives.

The last lemma required is \Cref{lem:np:if}, but this is completely analogous.
All in all, we obtain as for the undirected case:
\begin{corollary}
 \kprimsd is \NP-hard. 
\end{corollary}
\section{Approximation Algorithms}\label{sec:approxalgos}
The main part of this section consists in describing and analyzing a constant approximation algorithm for \minprimsu (\Cref{sec:apx:undirected}).
As it turns out, an algorithm for \minprimsd can be obtained by a suitable adaptation of this algorithm, as we will elaborate on in \Cref{subsec:apx:directed}.

As an ingredient our algorithms use a 2-approximation algorithm for the following problem:
\begin{definition}[Undirected Steiner Forest Problem (\USF)]
The \emph{Undirected Steiner Forest Problem} (\USF) is defined as follows:
For a given input $(G,S)$ such that $G$ is a graph and $S$ is a set of pairs of nodes from $G$, find a forest $F$ in $G$ with a minimum number of edges such that the two nodes of each pair in $S$ are connected by a path in $F$.
\end{definition}
As described in \Cref{subsec:related_work}, such an approximation algorithm already exists.
We will use it as a black box and thus not discuss it any further.

\subsection{A Constant Approximation Algorithm for \minprimsu}\label{sec:apx:undirected}
We now describe an approximation algorithm for \minprimsu (\Cref{subsec:apx:description}) and prove it to have a constant approximation ratio (\Cref{subsec:apx:analysis}).

\subsubsection{Algorithm Description}\label{subsec:apx:description}
For an initial graph $G_s = (V, E_s)$ and a final graph $G_t = (V, E_t)$, we define the set of \emph{additional} edges $\Eadditional := E_t \setminus E_s$ and the set of \emph{excess} edges $\Eremoved := E_s \setminus E_t$.
We now describe the algorithm in detail and then summarize its pseudo-code in \Cref{alg:apx}.
Our algorithm consists of two parts, the first of which deals with establishing all additional edges and the second of which deals with removing all excess edges.
In the first part, using an arbitrary 2-approximation algorithm for the \USF as a black box, the algorithm computes a 2-approximate solution to the following \USF instance:
The given graph is $G_s$, and the set of pairs of nodes is \Eadditional.
Note that the result is a forest such that for every edge $\{u,v\} \in \Eadditional$, $u$ and $v$ belong to the same tree.
For each tree $T$ in this forest the algorithm then selects an arbitrary root $r_T$ and connects all nodes in $T$ that are incident to an edge in $\Eadditional$ to $r_T$.
The exact details of this will be described when we analyze the length of the resulting computation.
In the next step, for every tree $T$, and every $\{u,v\} \in \Eadditional$ such that $u$ and $v$ belong to $T$, $r_T$ introduces $u$ and $v$ to each other, thereby creating the edge $\{u,v\}$.
After that, the superfluous edges (i.e., those edges that belong neither to $G_s$ nor to $\Eadditional$) are deleted in a bottom-up fashion:
every node that does not have a descendant with a superfluous edge (in the tree $T$ this node belongs to when viewing this tree as rooted by $r_T$), fuses all superfluous edges and delegates the last such to its parent in the tree.
Note that all superfluous edges in the same tree $T$ have $r_T$ as one of their endpoints.
The second part of the algorithm is similar to the first, with the following differences:
In the fifth step, the \USF is approximated for the input $(G_t, \Eremoved)$.
Note that the solution is a subgraph of the graph obtained after the first part of the algorithm.
In the sixth step, only one of the two endpoints of an edge from $\Eremoved$ is selected to become connected with the root of the tree the endpoints belong to.
In the seventh step (where in the first part the additional edges are created by the $r_T$ nodes), for each edge $e \in \Eremoved$, the endpoint selected in the sixth step delegates this edge to $r_T$ (resulting in the edge $\{r_T, v\}$).
These edges can then be delegated and fused in a bottom-up fashion by the endpoints other than $r_T$ in Step~8.
In contrast to Step~4, we begin with applying fusions here because the edges superfluous $\{r_T, v\}$ exist twice here (one was created in Step~6, and one in Step~7).

\begin{algorithm}[t]
  \caption{Approximation algorithm for \minprimsu
    \label{alg:apx}}
  \begin{algorithmic}[1]
  \Statex \textbf{Input:} Initial graph $G_s$ and final graph $G_t$.
  \Statex \par\smallskip
    \emph{First part (add additional edges):}
    \State Compute a 2-approximate solution \FALGplus for the \USF with input $(G_s,\Eadditional)$.%
    \State For each tree $T$ in \FALGplus, select a root node $r_T$ and connect all nodes in $T$ that are incident to an edge in $\Eadditional$ with $r_T$ (for details see the proof of \Cref{lem:apx:alg1_to_opt1}). %
    \State For each $\{u,v\} \in \Eadditional$, the root of the tree $u$ and $v$ belong to applies the introduction primitive to create the edge $\{u,v\}$. %
    \State For each tree $T$ in \FALGplus, delegate all superfluous edges (i.e., not belonging to $G_s$ or $\Eadditional$) created during Step~2 bottom up in $T$ rooted at $r_T$, starting with the lowest level. At each intermediate node fuse all of these edges before delegating them to the next parent.
    \Statex \par\smallskip
    \emph{Second part (remove excess edges):}
    \State Compute a 2-approximate solution \FALGminus for the \USF with input $(G_t, \Eremoved)$.
    \State For each $e \in \Eremoved$, let $s(e)$ be an arbitrary of the two endpoints of $e$.
    For each tree $T$ in \FALGminus, select a root node $r_T$ and for each $e \in \Eremoved$ whose endpoints belong to $T$, connect $s(e)$ with $r_T$ (similar to Step~2, for details see the proof of \Cref{lem:apx:alg2_to_opt2}).
    \State For each $e \in \Eremoved$, $s(e)$ delegates $e$ to $r_T$.
    \State For each tree $T$ in \FALGminus, delegate all superfluous edges (i.e., not belonging to $G_t$) bottom-up while fusing multiple edges as in Step~4.
  \end{algorithmic}
\end{algorithm}

\subsubsection{Analysis}\label{subsec:apx:analysis}
In this section we show that \Cref{alg:apx} is a constant-approximation algorithm for \minprimsu, which proves the following theorem:
\begin{theorem}\label{thm:minprimsu_in_apx}
 $\minprimsu \in \APX$.
\end{theorem}
For convenience we will analyze the two parts of the algorithm individually.
Therefore, for a given initial graph $G_s$ and final graph $G_t$, let $ALG_1(G_s,G_t)$ be the length of the computation of the first part of the algorithm for this instance, $ALG_2(G_s,G_t)$ be the length of the computation of the second part, and $ALG(G_s,G_t) := ALG_1(G_s,G_t) + ALG_2(G_s,G_t)$.
Furthermore, let $OPT(G_s,G_t)$ be the length of an optimal solution to \minprimsu for initial graph $G_s$ and final graph $G_t$.
We also define the intermediate graph $G' = (V, E_s \cup \Eadditional)$.
In the course of the analysis we will establish a relationship between $ALG_1(G_s,G_t)$ and $OPT(G_s, G')$ and between $ALG_2(G_s,G_t)$ and $OPT(G', G_t)$.
This will aid us in determining the approximation factor of \Cref{alg:apx} due to the following lemma:
\begin{lemma}\label{lem:optsum}
 $OPT(G_s, G') + OPT(G', G_t) \le 2OPT(G_s,G_t) + |\Eadditional|$.
\end{lemma}
\begin{proof}
 Let $\mathcal{P}$ denote the problem equal to \minprimsu with initial graph $G_s$ and final graph $G_t$ with the additional requirement that the computation must contain $G'$ %
 and let $OPT'(G_s, G', G_t)$ be the length of an optimal solution to it.
 Clearly, $OPT(G_s, G') + OPT(G', G_t) = OPT'(G_s, G', G_t)$ (note that one could split the optimal solution to $P$ at $G'$).
 We now show that $OPT'(G_s, G', G_t) \le 2OPT(G_s,G_t) + |\Eadditional|$.
 
 Consider a computation $C$ whose initial graph is $G_s$, whose final graph is $G_t$ and whose length is $OPT(G_s, G_t)$ (note that such a computation is an optimal solution to \minprimsu).
 We now transform $C$ into a computation that represents a solution to $\mathcal{P}$.
 This transformation increases its length by only $OPT(G_s,G_t) + |\Eadditional|$ and thus proves the above claim.
 
 Note that every edge $\{u,v\} \in \Eremoved$ is removed during $C$.
 This may happen due to the application of a fusion primitive (in case $G_t$ contains another edge; recall that edges may have a multiplicity of more than one) or a delegation primitive.
 For all edges to which the former case applies, we simply omit that application of a fusion primitive.
 For all edges to which the latter case applies, we replace each of these applications of delegation primitives by applications of introduction primitives.
 Altogether, we obtain a new computation $C'$ of equal length.
 Note that changing these primitive applications does not make the computation infeasible as this only causes the graph to have additional edges.
 The final graph of $C'$ is $(V, E_t \cup \Eremoved) = (V, E_s \cup \Eadditional) = G'$ (recall that $E_t = (E_s \cup \Eadditional) \setminus \Eremoved$).
 Next we append $C'$ by $C$ and obtain the computation $C''$ of length $2OPT(G_s, G_t)$.
 Note that since $C$ transformed $G_s$ to $G_t$, this second half of $C''$, which starts from $G' = (V, E_s \cup \Eadditional)$, has the final graph $G'' = (V, E_t \cup \Eadditional)$, i.e., each edge from $\Eadditional$ appears twice in $G''$.
 Thus we extend $C''$ by fusing each edge from $\Eadditional$ with its double, resulting in a computation $C'''$ of length $2OPT(G_s, G_t) + |\Eadditional|$.
 Since $C'''$ represents a solution to $\mathcal{P}$ for initial graph $G_s$ and final graph $G_t$, this completes the proof.
\end{proof}

One might question whether this bound is tight or if it is possible to, e.g., get rid of the additional $|\Eadditional|$ addend by a more careful analysis.
\Cref{fig:opt_tight} shows that this bound is indeed tight and shows an example where $OPT(G_s, G') + OPT(G', G_t) = 2OPT(G_s,G_t) + |\Eadditional|$.

\begin{figure}
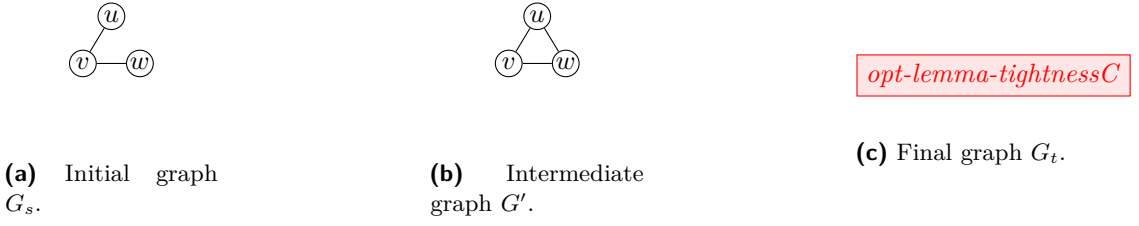

\centering
 \begin{subfigure}{.20\textwidth}
  \ctikzfig{opt-lemma-tightnessA}
  \caption{Initial graph $G_s$.}
 \end{subfigure}
\hfill
 \begin{subfigure}{.20\textwidth}
  \ctikzfig{opt-lemma-tightnessB}
  \caption{Intermediate graph $G'$.}
 \end{subfigure}
\hfill
 \begin{subfigure}{.20\textwidth}
  \ctikzfig{opt-lemma-tightnessC}
  \caption{Final graph $G_t$.}
 \end{subfigure}
 \caption{Example in which $OPT(G_s, G') + OPT(G', G_t) = 2OPT(G_s,G_t) + |\Eadditional|$: Transforming $G_s$ to $G'$ requires one application of an introduction primitive. To transform $G'$ into $G_t$, one needs to apply a delegation and a fusion. $OPT(G_s,G_t)$, on the other hand, is one since $v$ can simply delegate $\{u,v\}$ to $w$.}\label{fig:opt_tight} 
\end{figure}

In the rest of the analysis we show that $ALG_1(G_s,G_t) \le 11OPT(G_s, G')$ (\Cref{lem:apx:alg1_to_opt1}) and that $ALG_2(G_s,G_t) \le 11OPT(G', G_t)$ (\Cref{lem:apx:alg2_to_opt2}).
By Lemma~\ref{lem:optsum} this implies that $ALG(G_s,G_t) \le 11(2OPT(G_s,G_t) + |\Eadditional|) \le 33OPT(G_s,G_t)$ (since, clearly, $OPT(G_s,G_t) \ge |\Eadditional|$), which yields the claim of \Cref{thm:minprimsu_in_apx}.
We start with the former claim:
\begin{lemma}\label{lem:apx:alg1_to_opt1}
 $ALG_1(G_s,G_t) \le 11OPT(G_s, G')$.
\end{lemma}
\begin{proof}
Let \FOPTplus be an optimal solution to the \USF with input $(G_s, \Eadditional)$ and recall that $\FALGplus$ is the \USF approximation computed in Step~1 of \Cref{alg:apx}.
Throughout the analysis, $|\FOPTplus|$ and $|\FALGplus|$ will denote the number of edges in these solutions.
In the first part of this proof, we show that $ALG_1(G_s,G_t) \le 4|\FOPTplus| + 3|\Eadditional|$.
The second part then consists in proving $OPT(G_s, G') \ge |\FOPTplus| - |\Eadditional|$, which together with the observation that $OPT(G_s, G') \ge |\Eadditional|$ yields the claim.

To upper bound $ALG_1(G_s,G_t)$, we analyze the number of primitives applied in each of the steps of the first part of the approximation algorithm.
In Step~1, no primitive is applied.
To keep the number of edges as low as possible (which saves fusion primitives in Step~4), the algorithm for every $T$ in $\FALGplus$ connects the desired nodes to $r_T$ in Step~2 in the following way:
To simplify the description, we view $T$ as rooted at $r_T$ and for a node $u \in T$ denote by $ST(u)$ the set consisting of $u$ and all of its descendants in the tree $T$ rooted at $r_T$.
We say a node $u$ is \emph{relevant} if $ST(u)$ contains a node with an endpoint in $\Eadditional$\footnote{Note that although any tree in \FALGplus that contains nodes that are not relevant for any root could trivially be reduced in size, we have to take into account that such trees exist since we treat the approximation algorithm for \USF as a black box.}.
See \Cref{fig:algopicture} for an illustration of these notions.
First of all, $r_T$ introduces itself to all relevant children.
Then, starting from the second level, we proceed level-wise in the tree:
For each level $i$, every node $u$ at level $i$ checks whether $u$ is an endpoint of an edge in $\Eadditional$.
If so, $u$ introduces $r_T$ to all relevant children.
Otherwise, $u$ introduces $r_T$ to all but one of its relevant children (chosen arbitrarily) and delegates $r_T$ to the relevant child it did not introduce $r_T$ to.
The result of this procedure is that  each node incident to an edge in $\Eadditional$ has an edge to $r_T$ for the tree $T$ it belongs to%
, see \Cref{fig:algo1}.
Note that according to the definition of \FALGplus, for each pair $\{u,v\} \in \Eadditional$ $u$ and $v$ belong to the same tree $T$.
The above procedure increases the number of edges by at most $2|\Eadditional|$ and requires at most $|\FALGplus|$ applications of primitives (at most one for every edge in $T$).
It is easy to see that Step~3 (c.f.~\Cref{fig:algo2}) involves exactly $\Eadditional$ applications of primitives.
For the length of Step~4  (c.f.~\Cref{fig:algo3}), note that for every tree $T$ at most one delegation has to be applied for every edge in $T$ (causing $|\FALGplus|$ delegations in total) and at most $2|\Eadditional|$ fusions have to be applied for this is the number of superfluous edges created during Step~2.
All in all, Step~2, Step~3, and Step~4 involve $|\FALGplus|$, $|\Eadditional|$, and $|\FALGplus| + 2|\Eadditional|$ applications of primitives, respectively.
This makes a total of $2|\FALGplus| + 3|\Eadditional|$.
Since $\FALGplus$ is a 2-approximation of $\FOPTplus$, we obtain $ALG_1(G_s,G_t) \le 4|\FOPTplus| + 3|\Eadditional|$.

\begin{figure}
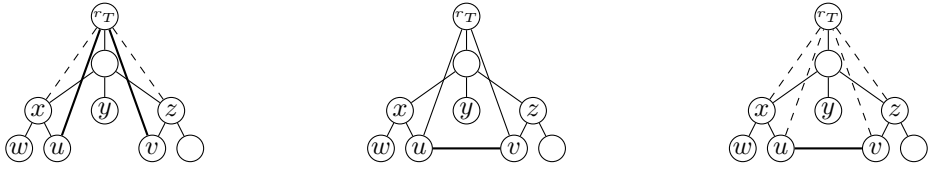

\centering
 \begin{subfigure}{.32\textwidth}
  \ctikzfig{algo-pictureA}
  \caption{Step~2 connects all endpoints of edges in $\Eadditional$ belonging to $T$ with $r_T$.}
  \label{fig:algo1}
 \end{subfigure}
\hfill
 \begin{subfigure}{.32\textwidth}
  \ctikzfig{algo-pictureB}
  \caption{In Step~3, $r_T$ creates the edges in $\Eadditional$ that belong to $T$ by an introduction.}
  \label{fig:algo2}
 \end{subfigure}
\hfill
 \begin{subfigure}{.32\textwidth}
  \ctikzfig{algo-pictureC}
  \caption{Step~4 removes all superfluous edges by delegating and fusing them up in the tree.}
  \label{fig:algo3}
 \end{subfigure}
 \caption{Example of a tree $T$ with root $r_T$ for Step~2-4 of \Cref{alg:apx} assuming $\{u,v\} \in \Eadditional$. $ST(x)$ consists of $x$, $w$, and $u$. $x$ is relevant, whereas $y$ is not. Dashed edges exist temporarily during the displayed step.}\label{fig:algopicture}
\end{figure}

For the lower bound on $OPT(G_s, G')$, assume for contradiction that there is a computation $C$ with initial graph $G_s$ and final graph $G'$ of length $L < |\FOPTplus| - |\Eadditional|$.
Let $G_s = G_1 \Rightarrow G_2 \Rightarrow \dots \Rightarrow G_L$ be the sequence of graphs of this computation.
For every $\{u,v\} \in \Eadditional$ we iteratively create a path from $u$ to $v$ in the following way:
Begin with $P_{u,v}^L:=(u,v)$.
Note that $P_{u,v}^L$ exists in $G_L$.
We iterate through $C$ in reverse order and for every graph $G_i$, if $P_{u,v}^{i+1}$ exists in $G_i$, $P_{u,v}^{i} := P_{u,v}^{i+1}$.
Otherwise, since $G_{i+1}$ is the result of a single application of a primitive to $G_i$, there is exactly one edge $\{x,y\}$ in $P_{u,v}^{i+1}$ that exists in $G_{i+1}$ but not in $G_i$ and this edge was created by the application of an introduction or delegation primitive of some node $w$ such that $\{w,x\}$ and $\{w,y\}$ exist in $G_i$.
Thus, let $P_{u,v}^{i}$ be $P_{u,v}^{i+1}$ with $(x,y)$ replaced by $(x,w,y)$ and note that $P_{u,v}^{i}$ exists in $G_i$.
Eventually, we obtain a path $P_{u,v}^{1}$ that exists in $G_s$.
For $i \in \{1, \dots, L\}$, let $F^i := \bigcup_{\{u,v\} \in \Eadditional}E(P_{u,v}^{i})$ (where $E(P)$ is the set of all edges on the path $P$) be a normal set of edges (i.e., not a multiset).
Note that $F^1$ can be transformed into a feasible (though not necessarily optimal) solution to the \USF with input $(G_s,\Eadditional)$ by removing cycle edges.
Therefore, $|F^1| \ge |\FOPTplus|$.
An example is given in \Cref{fig:algoforest}.
\begin{figure}
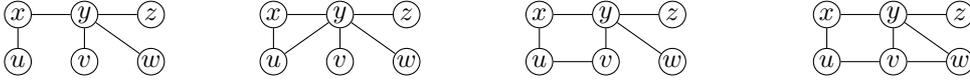
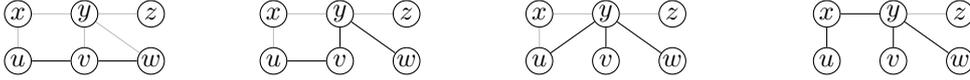

\centering
\captionsetup[subfigure]{aboveskip=0pt}
 \begin{subfigure}{.22\textwidth}
  \ctikzfig{algoproof-forestA}
  \caption{Initial graph $G_1$.\\}
  \label{fig:algoforest1}
 \end{subfigure}
\hfill
 \begin{subfigure}{.22\textwidth}
  \ctikzfig{algoproof-forestB}
  \caption{$G_2$: $x$ has introduced $u$ to $y$.}
  \label{fig:algoforest2}
 \end{subfigure}
\hfill
 \begin{subfigure}{.24\textwidth}
  \ctikzfig{algoproof-forestC}
  \caption{$G_3$: $y$ has delegated $u$ to $v$.}
  \label{fig:algoforest3}
 \end{subfigure} 
\hfill
 \begin{subfigure}{.26\textwidth}
  \ctikzfig{algoproof-forestD}
  \caption{$G_4$: $y$ has introduced $v$ to $w$.}
  \label{fig:algoforest4}
 \end{subfigure}
 
 \par\bigskip 
 \begin{subfigure}{.22\textwidth}
  \ctikzfig{algoproof-forestDb}
  \caption{$G_4$: $P_{u,v}^4 = (u,v)$ and $P_{v,w}^4 = (v,w)$.}
  \label{fig:algoforest4b}
 \end{subfigure}   
\hfill 
\begin{subfigure}{.22\textwidth}
  \ctikzfig{algoproof-forestCb}
  \caption{$G_3$: $P_{u,v}^3 = (u,v)$ and $P_{v,w}^3 = (v,y,w)$.}
  \label{fig:algoforest3b}
 \end{subfigure}    
\hfill 
\begin{subfigure}{.24\textwidth}
  \ctikzfig{algoproof-forestBb}
  \caption{$G_2$: $P_{u,v}^2 = (u,y,v)$ and $P_{v,w}^2 = (v,y,w)$.}
  \label{fig:algoforest2b}
 \end{subfigure}     
 \hfill 
\begin{subfigure}{.26\textwidth}
  \ctikzfig{algoproof-forestAb}
  \caption{$G_1$: $P_{u,v}^1 = (u,x,y,v)$ and $P_{v,w}^1 = (v,y,w)$.}
  \label{fig:algoforest1b}
 \end{subfigure}     
 \caption{Example of an optimal computation $C$ with initial graph $G_1$ and $\Eadditional = \{ \{u,v\}, \{v,w\} \}$, and the notions used in the proof of \Cref{lem:apx:alg1_to_opt1}.
 The upper row shows $C$ in order, the lower row illustrates the path sets $P_{u,v}^i$ and $P_{v,w}^i$, which are defined by iterating through $C$ in reverse order.
 In the lower row, the edges drawn black in $G_i$ are the edges belonging to to $F^i$.
 Observe that $F^1$ is a superset of a feasible solution to the \USF for graph $G_1$ and node pairs $\Eadditional$.}\label{fig:algoforest}
\end{figure}
For an arbitrary $i \in \{1, \dots, L-1\}$, note that $|F^i| \le |F^{i+1}|+1$:
if $G_{i+1}$ was obtained from $G_i$ by the application of a fusion primitive, this inequality trivially holds as none of the above paths changes in this case.
Otherwise, $G_{i+1}$ was obtained from $G_i$ by an application of an introduction or delegation primitive by some node $w$ causing at most one edge $\{x,y\}$ to exist in $G_{i+1}$ that does not exist in $G_i$.
In this case, we further know that $\{w,x\}$ and $\{w,y\}$ exist in $G_i$ and by the definition of the above paths, for every pair $\{u,v\}$ such that $P_{u,v}^{i+1}$ contains the edge $\{x,y\}$ the path $P_{u,v}^i$ contains $(x,w,y)$ as a sub-path instead and for all other pairs $\{u',v'\}$, $P_{u',v'}^i = P_{u',v'}^{i+1}$.
By the definition of $F^i$ and $F^{i+1}$, this implies $|F^i| \le |F^{i+1}|+1$ also in this case.
All in all we obtain that $|F^1| \le |F^L| + L = |\Eadditional| + L$ because $F^L = \Eadditional$ (note the definition of $F^L$).
By the assumption that $L < |\FOPTplus| - |\Eadditional|$, we obtain $|F^1| < |\FOPTplus|$, which represents a contradiction.
\end{proof}

\begin{lemma}\label{lem:apx:alg2_to_opt2}
 $ALG_2(G_s,G_t) \le 11OPT(G', G_t)$.
\end{lemma}
\begin{proof}
The general structure of this proof follows the line of the proof of \Cref{lem:apx:alg1_to_opt1}, but differs in the details.
Similar to the notation used in that proof, let \FOPTminus be an optimal solution for the \USF with input $(G_t,\Eremoved)$ and recall that $\FALGminus$ is the \USF approximation computed in Step~5 of \Cref{alg:apx}.
Analogously, $|\FOPTminus|$ and $|\FALGminus|$ denote the number of edges in these solutions.
In the first part of this proof, we show that $ALG_2(G_s,G_t) \le 4|\FOPTminus| + 3|\Eremoved|$.
The second part then consists in proving $OPT(G', G_t) \ge |\FOPTminus| - |\Eremoved|$, which together with the observation that $OPT(G', G_t) \ge |\Eremoved|$ yields the claim.

To upper bound $ALG_2(G_s,G_t)$, we analyze the number of primitives applied in each step of the second part of the approximation algorithm.
Of course, no primitive is applied in Step~5.
The connections required in Step~6 can be created in a similar fashion as in Step~2, which is described in the  proof of \Cref{lem:apx:alg1_to_opt1}:
For each tree $T$, we proceed top-down in $T$ rooted at some arbitrary but fixed node $r_T$ again.
Here, each intermediate node $u$ checks whether $u = s(e)$ for some $e \in \Eremoved$.
If so, it introduces $r_T$ to all relevant children (here a node $v$ is \emph{relevant} if $ST(v)$ contains a node $w$ such that $w = s(e')$ for some $e' \in \Eremoved$).
Otherwise, it introduces $r_T$ to all but one relevant children and delegates it to the remaining one.
In the end, for every edge $e \in \Eremoved$, $s(e)$ has an edge to $r_T$, the number of edges in the graph has increased by at most $|\Eremoved|$, and the process involved at most $|\FALGminus|$ applications of primitives.
In Step~7, clearly exactly $|\Eremoved|$ edges have to be delegated.
Step~8 is similar to Step~4 and for analogous reasons requires at most $|\FALGminus|$ delegations and at most $2|\Eremoved|$ fusions (recall that up to $|\Eremoved|$ edges were added in Step~6 and the edges delegated in Step~7 have to be removed as well).
All in all, Step~6, Step~7 and Step~8 of the algorithm involve at most $|\FALGminus|$, $|\Eremoved|$ and $|\FALGminus| + 2|\Eremoved|$ applications of primitives, respectively, which yields: $ALG_2(G_s,G_t) \le 2|\FALGminus| + 3 |\Eremoved| \le 4|\FOPTminus| + 3 |\Eremoved|$ (since \FALGminus is a 2-approximation of \FOPTminus).

To lower bound the value of $OPT(G', G_t)$, assume for contradiction that there is a computation $C$ with initial graph $G'$ and final graph $G_s$ of length $L < |\FOPTminus| - |\Eremoved|$.
Let $G_s = G_1 \Rightarrow G_2 \Rightarrow \dots \Rightarrow G_L$ be the sequence of graphs of this computation.
Similar to the proof of \Cref{lem:apx:alg1_to_opt1}, for every $\{u,v\} \in \Eremoved$, we create a path from $u$ to $v$, but this time we start with $P_{u,v}^1 := (u,v)$ and consider the graphs in increasing order:
For $i \in \{2, \dots, L\}$, if $P_{u,v}^{i-1}$ exists in $G_i$, $P_{u,v}^{i} := P_{u,v}^{i-1}$.
Otherwise since $G_{i}$ is the result of a single application of a primitive to $G_{i-1}$, there is exactly one edge $\{x,y\}$ in $P_{u,v}^{i-1}$ that exists in $G_{i-1}$ but not in $G_i$ and this edge must have been delegated by either $x$ or $y$ to some node $w$.
In the following denote the node that applied the delegation by $z$ and denote by $\overline z$ the other node from $\{x,y\}$.%
In $G_{i-1}$, $z$ must share an edge with $w$ and this edge still exists in $G_i$ (for only one primitive is applied in the transition from $G_{i-1}$ to $G_i$).
Since $\{z, \overline z\}$ was delegated by $z$ to $w$, the edge $\{ w, \overline z\}$ exists in $G_i$.
Thus, let $P_{u,v}^i$ be $P_{u,v}^{i-1}$ with $(x,y)$ replaced by $(x,w,y)$ and observe that $P_{u,v}^i$ exists in $G_i$.
Eventually, we obtain a path $P_{u,v}^L$ that exists in $G_t$.
Define $F^i := \bigcup_{\{u,v\} \in \Eremoved}E(P_{u,v}^{i})$ (where $E(P)$ is the set of all edges on the path $P$) as a normal set of edges (i.e., not a multiset) for every $i \in \{1, \dots, L\}$, and note that $F^L$ can be transformed into a feasible (though not necessarily optimal) solution to the \USF with input $(G_t,\Eremoved)$ by removing cycle edges.
Therefore, $|F^L| \ge |\FOPTminus|$.
Furthermore, for an arbitrary $i \in \{1, \dots, L-1\}$, note that $|F^{i+1}| \le |F^{i}|+1$ because there is at most one edge $\{x,y\}$ that exists in $G_{i}$ but not in $G_{i+1}$ and thus causes the replacement of $(x,y)$ by $(x,w,y)$ for some fixed node $w$ for all paths that contain $(x,y)$ as a sub-path.
This yields that $|F^L| \le |F^1| + L = |\Eremoved| + L$ because $F^1 = \Eremoved$ (note the definition of $F^1$).
By the assumption that $L < |\FOPTminus| - |\Eremoved|$, we obtain $|F^L| < |\FOPTminus|$, which represents a contradiction.
\end{proof}

\subsection{A Constant Approximation Algorithm for \minprimsd}\label{subsec:apx:directed}
In this subsection we describe how to adapt \Cref{alg:apx} to obtain a constant approximation algorithm for \minprimsd.
The pseudocode of the adapted version of the approximation algorithm is given in \Cref{alg:apx_d} (with differences to \Cref{alg:apx} highlighted in boldface). 
In this pseudocode and in the following proof, for a graph $G = (V,E)$, $U(G)$ denotes the undirected version of $G$, i.e., each edge $(u,v)$ in $G$ is replaced by $\{u,v\}$.
An example of the procedure of \Cref{alg:apx_d} is illustrated in \Cref{fig:algopictureDirected}.

\begin{algorithm}[ht]
  \caption{Approximation algorithm for \minprimsd
    \label{alg:apx_d}}
  \begin{algorithmic}[1]
  \Statex \textbf{Input:} Initial graph $G_s$ and final graph $G_t$.
  \Statex \par\smallskip
    \emph{First part (add additional edges):}
    \State Compute a 2-approximate solution \FALGplus for the \USF with input $(\mathbf{U(G_s)}, \mathbf{U(\Eadditional)})$.
    \State For each tree $T$ in \FALGplus, select a root node $r_T$, \textbf{reverse all edges in $\mathbf{T}$ that are oriented towards $\mathbf{r_T}$} and for all nodes $u$ in $T$ that are incident to an edge in $\Eadditional$, establish the edge $\mathbf{(r_T,u)}$ (for details see the proof of \Cref{thm:minprimsd_in_apx}).
    \State For each $\mathbf{(u,v) \in \Eadditional}$, the root of the tree $u$ and $v$ belong to applies the introduction primitive to create the edge $\mathbf{(u,v)}$. %
    \State For each tree $T$ in \FALGplus, \textbf{reverse all edges in $\mathbf{T}$}, \textbf{reverse all superfluous edges (i.e., not belonging to $G_s$ or $\Eadditional$) created during Step~2} and delegate them bottom up in $T$ rooted at $r_T$, starting with the lowest level. At each intermediate node fuse all of these edges before delegating them to the next parent. \textbf{Afterwards, reverse all edges in $\textbf{T}$ that were originally oriented away from $\mathbf{r_T}$ (i.e., the edges in $\mathbf{T}$ that were not reversed in Step~2).}
    \Statex \par\smallskip
    \emph{Second part (remove excess edges):}
    \State Compute a 2-approximate solution \FALGminus for the \USF with input $(\mathbf{U(G_t)}, \mathbf{U(\Eremoved)})$.
    \State For each tree $T$ in \FALGminus, select a root node $r_T$, \textbf{reverse all edges in $\mathbf{T}$ that are oriented towards $\mathbf{r_T}$} and for each $\mathbf{(u,v) \in \Eremoved}$ whose endpoints belong to $T$, establish the edge $\mathbf{(u,r_T)}$ (similar to Step~2, for details see the proof of \Cref{thm:minprimsd_in_apx}).
    \State For each $\mathbf{(u,v) \in \Eremoved}$, $\mathbf{u}$ delegates $\mathbf{(u,v)}$ to $r_T$.
    \State For each tree $T$ in \FALGminus, \textbf{reverse all edges in $\mathbf{T}$}, \textbf{reverse all edges created during Step~7}, delegate all superfluous edges (i.e., not belonging to $G_t$), bottom-up while fusing multiple edges as in Step~4. \textbf{Afterwards, reverse all edges in $\textbf{T}$ that were originally oriented away from $\mathbf{r_T}$ (i.e., the edges in $\mathbf{T}$ that were not reversed in Step~6).}
  \end{algorithmic}
\end{algorithm}

\begin{figure}[ht]
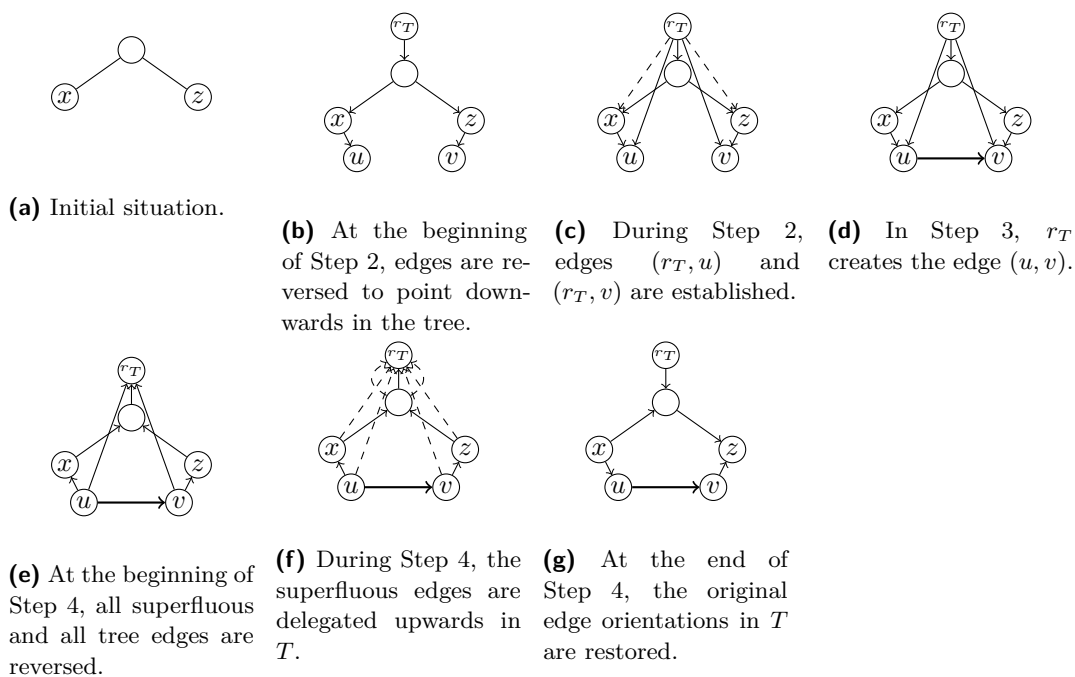
 
\centering
 \begin{subfigure}{.23\textwidth}
  \ctikzfig{algoDirected-pictureA}
  \caption{Initial situation. \\ \\ \\}
  \label{fig:algoDirected1}
 \end{subfigure}
\hfill
 \begin{subfigure}{.23\textwidth}
  \ctikzfig{algoDirected-pictureB}
  \caption{At the beginning of Step~2, edges are reversed to point downwards in the tree.}
  \label{fig:algoDirected2}
 \end{subfigure}
\hfill
 \begin{subfigure}{.23\textwidth}
  \ctikzfig{algoDirected-pictureC}
  \caption{During Step~2, edges $(r_T,u)$ and $(r_T,v)$ are established.\\}
  \label{fig:algoDirected3}
\end{subfigure}  
\hfill
   \begin{subfigure}{.23\textwidth}
  \ctikzfig{algoDirected-pictureD}
  \caption{In Step~3, $r_T$ creates the edge $(u,v)$.\\ \\}
  \label{fig:algoDirected4}
 \end{subfigure}

  \begin{subfigure}{.23\textwidth}
  \ctikzfig{algoDirected-pictureE}
  \caption{At the beginning of Step~4, all superfluous and all tree edges are reversed.}
  \label{fig:algoDirected5}
 \end{subfigure}
\hfill
 \begin{subfigure}{.23\textwidth}
  \ctikzfig{algoDirected-pictureF}
  \caption{During Step~4, the superfluous edges are delegated upwards in $T$.\\}
  \label{fig:algoDirected6}  
 \end{subfigure}
 \hfill
\begin{subfigure}{.23\textwidth}
  \ctikzfig{algoDirected-pictureG}
  \caption{At the end of Step~4, the original edge orientations in $T$ are restored.\\}
  \label{fig:algoDirected7}  
 \end{subfigure} 
 \hfill
\phantom{
\begin{subfigure}{.23\textwidth}
  \ctikzfig{algoDirected-pictureG}
  \caption{\\ \\ \\ \\}
 \end{subfigure} 
 }
 
 \caption{Example of the first part of \Cref{alg:apx_d} for a tree $T$ with root $r_T$ assuming $(u,v) \in \Eadditional$. Dashed edges exist temporarily during the displayed step.}\label{fig:algopictureDirected}
\end{figure}

\Cref{alg:apx_d} also has a constant approximation factor, which is stated by the following theorem:
\begin{theorem}\label{thm:minprimsd_in_apx}
 $\minprimsd \in \APX$.
\end{theorem}
\begin{proof}
We begin with establishing a relationship between solutions to \minprimsd and to \minprimsu.
For arbitrary directed graphs $G_s$ and $G_t$, let $C_d(G_s,G_t)$ be an optimal solution to \minprimsd with (directed) initial graph $G_s$ and (directed) final graph $G_t$ (this problem we denote by $P_d(G_s,G_t)$) and let $OPT_d(G_s,G_t)$ be its length.
Let $C_u((U(G_s),U(G_t))$ be an optimal solution to \minprimsu with initial graph $U(G_s)$ and final graph $U(G_t)$, denote its length by $OPT_u((U(G_s),U(G_t))$, and denote this problem by $P_u((U(G_s),U(G_t))$.
Let the computation $C_u'$ be obtained from $C_d$ by ignoring all edge directions and removing all applications of reversal primitives from $C_d$.
Note that $C_u'$ is a solution to $P_u((U(G_s),U(G_t))$ and that its length is thus lower bounded by $OPT_u((U(G_s),U(G_t))$ and upper bounded by $OPT_d(G_s,G_t)$ (since we only shortened the optimal solution to $P_d(G_s,G_t)$).
Therefore, $OPT_u((U(G_s),U(G_t)) \le OPT_d(G_s,G_t)$ for all directed graphs $G_s$ and $G_t$.
We will use this insight to compare the length of the solution computed by \Cref{alg:apx_d} for initial graph $G_s$ and final graph $G_t$ with the length of the solution computed by \Cref{alg:apx} for initial graph $U(G_s)$ and final graph $U(G_t)$ in order to prove the claim. 

In the following, due to the similarities of \Cref{alg:apx_d} with \Cref{alg:apx}, we only describe the changes to the algorithm that require additional explanation (namely, Step~2 and Step~6).
The other changes (highlighted in boldface in \Cref{alg:apx_d}) are self-explanatory.
Additionally, we analyze the influence of the differences in the algorithms on the length of the computations.

The procedure for the creation of edges in Step~2 is, in general, very similar to that used in \Cref{alg:apx} (described in the proofs of \Cref{lem:apx:alg1_to_opt1}).
There are three differences, though:
First, in order to delegate / introduce edges downwards in a tree $T$, all edges of $T$ must be oriented downwards, which is why we reverse all edges oriented into the opposite direction (towards $r_T$) first.
Let $E_R$ be the set of all these edges of all trees $T$.
Then this additionally involes $|E_R|$ applications of the reversal primitive in total.
Second, since we deal with directed edges, we have to clarify that the edges delegated / introduced through the tree must be directed towards the root node.
Third, this process creates the additional edges $(u,r_T)$ for every $u$ incident to an edge in $\Eadditional$, whereas $(r_T,u)$ is desired.
Thus, for each such edge, we apply the reversal primitive to turn $(r_T,u)$ into $(u,r_T)$, which requires at most $2|\Eadditional|$ applications of primitives.
In Step~4, we additionally have to reverse all edges in every tree first in order to be able to delegate upwards in the tree.
This involves $|\FALGplus|$ additional applications of reversal primitives in total.
Furthermore, we have to reverse all superfluous edges before we can start the bottom-up procedure.
As their number is, again, upper bounded by $2|\Eadditional|$, we have to apply only that number of reversal primitives.
In addition, we afterwards reverse all edges whose orientation does not yet match their orientation in $G_s$, which totally requires an additional $|\FALGplus \setminus E_R|$ applications of reversal primitives (since these are exactly the edges in \FALGplus that were not reversed in Step~2).
Altogether, the additional number of primitives required in comparison to \Cref{alg:apx} is $|E_R| + |2\Eadditional| + |\FALGplus| + 2|\Eadditional| + |\FALGplus \setminus E_R| = 2|\FALGplus| + 4|\Eadditional|$
Thus, the length $ALG_1^2(G_s,G_t)$ of the computation computed in the first part by \Cref{alg:apx_d} is at most $ALG_1^1(U(G_s),U(G_t)) + 4|\Eadditional| + 2|\FALGplus|$ (in which $ALG_1^1(U(G_s),U(G_t))$ is the length of the computation of the first part of \Cref{alg:apx} for initial graph $U(G_s)$ and final graph $U(G_t)$).

For the second part, note that the creation of edges in Step~6 is generally very similar to that used in \Cref{alg:apx} (described in the proof of \Cref{lem:apx:alg2_to_opt2}).
Here, for each edge $(u,v) \in \Eremoved$, $u$ corresponds to $s(e)$ in \Cref{alg:apx}.
As in Step~2, we have to reverse the edges downwards in the tree first, incuring $r$ additional primitive applications.
Again, the edges delegated / introduce through the tree must be directed towards the root node.
In Step~8, we again have to reverse all edges in each tree, requiring an additional number of primitive applications of $|\FALGminus|$.
Additionally, we the edges created in Step~7 need to be reversed, which are $|\Eremoved|$ in total.
Last, the edges whose orientation does not match their orientation in $G_t$ (which are those that were not reversed in Step~6) need to be reversed, which requires $|\FALGminus| - r$ applications of reversal primitives.
All in all, the length $ALG_2^2(G_s,G_t)$ of the computation computed in the second part by \Cref{alg:apx_d} is at most $ALG_2^1(G_s,G_t) + 2|\FALGminus| + |\Eremoved|$ (in which $ALG_2^1(U(G_s),U(G_t))$ is the length of the computation of the second part of \Cref{alg:apx} for initial graph $U(G_s)$ and final graph $U(G_t)$).

Incorporating the arguments and results of the proofs of \Cref{lem:apx:alg1_to_opt1} and \Cref{lem:apx:alg2_to_opt2}, we obtain:
\begin{align*}
ALG_1^2(G_s,G_t) &\le ALG_1^1(U(G_s),U(G_t)) + 4|\FOPTplus| + 4|\Eadditional| \\
 & \le 8|\FOPTplus| + 7|\Eadditional| \\
 & \le 8 OPT(U(G_s),U(G')) + 15|\Eadditional| \\
 & \le 23 OPT(U(G_s),U(G')), \text{ and}\\
ALG_2^2(G_s,G_t) &\le ALG_2^1(U(G_s),U(G_t)) + 4|\FOPTminus| + |\Eremoved| \\
 & \le 8|\FOPTminus|+ 4|\Eremoved|\\
 & \le 8 OPT(U(G'),U(G_t)) + 12|\Eremoved| \\
 & \le 20 OPT(U(G'),U(G_t)). 
\end{align*}
where $OPT$ is defined as in \Cref{subsec:apx:analysis} and $G' = U((V, E_s \cup \Eadditional))$ for $E_s$ being the edge set of $G_s$.

Let $ALG^2(G_s,G_t)$ be the length of the computation computed by \Cref{alg:apx_d} for initial graph $G_s$ and final graph $G_t$.
Building on the above two inequations, we can apply \Cref{lem:optsum} to obtain
\begin{align*}
ALG^2(G_s,G_t) :=& ALG_1^2(G_s,G_t) + ALG_2^2(G_s,G_t) \\
&\le 23(2OPT_u(U(G_s),U(G_t)) + |\Eadditional|) \\
& \le 69OPT_u(U(G_s),U(G_t)) \\
& \le 69OPT_d(G_s,G_t).
\end{align*}
This finishes the proof.
\end{proof}
\section{Conclusion}\label{sec:conclusion}
We proposed a set of primitives for topology adaptation that a server may use to adapt the network topology into any desired (weakly) connected state but at the same time cannot use to disconnect the network or to introduce new nodes into the system.
So far, we only assumed that the server could act maliciously but that the participants of the network are honest and correct, i.e., they refuse any graph transformation commands beyond the four primitives.
What, however, if some participants also behave in a malicious manner? Is it still possible to avoid Eclipse or Sybil attacks?
It seems that in this case the only measure that would help is to form quorums of nodes that are sufficiently large so that at least one node in each quorum is honest.

Besides these security-related aspects, our results give rise to additional questions:
For example, does the NP-hardness apply to any set of local primitives, or is there a set of local primitives that can transform arbitrary initial graphs much faster into arbitrary final graphs than the set considered in this work?
Furthermore, is it possible to obtain decentralized versions of the algorithms presented in \Cref{sec:approxalgos}, and, if so, what is their competitiveness when compared to the centralized ones?  
\bibliographystyle{plainurl}
\bibliography{literature.bib}  %

\end{document}